\documentclass[referee,sn-mathphys-num]{sn-jnl}

\usepackage{mathptmx}

\usepackage{graphicx}%
\usepackage{multirow}%
\usepackage{amsmath,amssymb,amsfonts}%
\usepackage{amsthm}%
\usepackage{mathrsfs}%
\usepackage[title]{appendix}%
\usepackage{xcolor}%
\usepackage{textcomp}%
\usepackage{manyfoot}%
\usepackage{booktabs}%
\usepackage{algorithm}%
\usepackage{algorithmicx}%
\usepackage{algpseudocode}%
\usepackage{listings}%
\usepackage{braket}
\usepackage{mathtools}
\usepackage{enumitem}
\usepackage{dsfont}

\usepackage{cleveref}
\usepackage{verbatim}

\theoremstyle{thmstyleone}%
\newtheorem{theorem}{Theorem}
\newtheorem{proposition}[theorem]{Proposition}%
\newtheorem{lemma}[theorem]{Lemma}

\theoremstyle{thmstyletwo}%

\theoremstyle{remark}
\newtheorem*{remark*}{Remark}%

\theoremstyle{thmstylethree}%
\newtheorem{definition}[theorem]{Definition}%

\DeclareMathOperator{\Tr}{Tr}
\newcommand{\abs}[1]{\left\vert#1\right\vert}
\newcommand{\defeq}{\vcentcolon=}

\crefname{enumi}{condition}{conditions}
\crefname{proposition}{Proposition}{Proposition}
\crefname{lemma}{Lemma}{Lemmas}
\crefname{definition}{Definition}{Definitions}

\begin{document}

\title[Equivalence of MUBs via orbits]{Equivalence of mutually unbiased bases via orbits: general theory and a $d=4$ case study}


\author*[1]{\fnm{Amit} \sur{Te'eni}}\email{amit.teeni@biu.ac.il}

\author[1]{\fnm{Eliahu} \sur{Cohen}}

\affil*[1]{\orgdiv{Faculty of Engineering and Institute of Nanotechnology and Advanced Materials}, \orgname{Bar Ilan University}, \orgaddress{\city{Ramat Gan}, \postcode{5290002}, \country{Israel}}}


\abstract{In quantum mechanics, mutually unbiased bases (MUBs) represent orthonormal bases that are as ``far apart'' as possible, and their classification reveals rich underlying geometric structure.
Given a complex inner product space, we construct the space of its orthonormal bases as a discrete quotient of the complete flag manifold. We introduce a metric on this space, which corresponds to the ``MUBness'' distance. This allows us to describe equivalence between sets of mutually unbiased bases in terms of the geometry of this space. The subspace of bases that are unbiased with respect to the standard basis decomposes into orbits under a certain group action, and this decomposition corresponds to the classification of complex Hadamard matrices. More generally, we consider a list of $k$ MUBs, that one wishes to extend. The candidates are points in the subspace comprising all bases which are unbiased with respect to the entire list. This space also decomposes into orbits under a group action, and we prove that points in distinct orbits yield inequivalent MUB lists. Thus, we generalize the relation between complex Hadamard matrices and MUBs.
As an application, we identify new symmetries that reduce the parameter space of MUB triples in dimension $4$ by a factor of $4$.
}




\maketitle

\section{Introduction}\label{intro}
Consider the inner product (Hilbert) space $ \mathcal{H} = \mathbb{C}^n $, and let $ p = \left\{ \ket{\phi_i} \right\}_{i=1}^n $ and $ q = \left\{ \ket{\psi_j} \right\}_{j=1}^n $ be two orthonormal bases (ONBs). If the inner products $ \braket{ \phi_i \vert \psi_j } $ all have the same modulus, then $p$ and $q$ are said to be \textit{mutually unbiased}. In this case the modulus must be $ 1 / \sqrt{n} $, i.e.:
\begin{equation}\label{MUB_def}
	\forall i, j, \quad \abs{ \braket{ \phi_i \vert \psi_j } }^2 = \frac{1}{n} .
\end{equation}
The term ``mutually unbiased bases'' (MUBs) originates in quantum information theory. Suppose $ \mathcal{H} $ is the Hilbert space of some quantum system, and consider two consecutive projective measurements of the state: first with respect to the basis $p$, and then with respect to $q$. In this scenario, the outcome of the first measurement provides no information regarding the outcome of the second one; the probability distribution for the second measurement is always uniform. In this sense, the two measurements are mutually unbiased.
Similarly, a set of orthonormal bases $ \left\{ q_k \right\}_{k=1}^m $ is referred to as a set of MUBs, if any two distinct bases $ q_k, q_l $ are mutually unbiased. It can be shown that no such set can contain more than $n+1$ bases~\cite{wootters1989optimal}. A set of exactly $n+1$ MUBs is called \textit{complete}.
Mutually unbiased bases have numerous applications in quantum information theory. A complete set of MUBs yields an optimal scheme for quantum state determination~\cite{ivonovic1981geometrical,wootters1989optimal,PhysRevLett.105.030406}.
MUBs are also used in quantum key distribution (QKD)~\cite{PhysRevLett.67.661,PhysRevLett.81.3018,PhysRevLett.88.127902,PhysRevLett.92.057901,RevModPhys.81.1301,pirandola2020advances,ashkenazy2024photon}, random access codes~\cite{PhysRevA.78.022310,PhysRevLett.114.170502}, quantum tomography~\cite{PhysRevA.85.052120,Lima:11}, entropic uncertainty relations~\cite{PhysRevA.91.042133,PhysRevA.104.062204,rastegin2013uncertainty,PhysRevA.79.022104}, and entanglement detection~\cite{PhysRevA.86.022311,tavakoli2021mutually,hiesmayr2021detecting}.

Due to the immense fundamental and practical significance of MUBs, it is important to understand their properties. One puzzle that has attracted the attention of many researchers, pertains to the largest possible set of MUBs.
If $n$ is a prime power, then a complete set of MUBs always exists; but in composite dimensions, it is generally unknown whether a complete set exists or not~\cite{brierley2009constructing,weiner2013gap,horodecki2022five,mcnulty2024mutually}. Although the problem is formulated in simple terms of basic linear algebra, it turns out to be a fairly deep one, related to a myriad of fields in physics and mathematics alike~\cite{planat2006survey,bengtsson2007three,bengtsson2007mutually,Boykin2007Mutually,brierley2009mutually,brierley2009all,godsil2009equiangular,durt2010mutually}.
A separate (yet related) problem is the classification of MUB sets up to equivalence. Two ordered lists $ \left( p_1, \ldots, p_k \right) $ and $ \left( q_1, \ldots, q_k \right) $ of MUBs are \textit{unitarily equivalent} is there exists a unitary $U$ such that $ U $ maps $p_i$ to $q_i$ for all $1 \leq i \leq k$. Inequivalent MUB sets can have distinct formal properties, and may also yield different performance in the aforementioned applications (see Section 3.13 of~\cite{mcnulty2024mutually}). Note that equivalence of MUB triples was recently studied theoretically~\cite{matolcsi2025triplets} and experimentally~\cite{PhysRevLett.132.080202}.

In this paper, we study the classification of MUB sets via group actions on manifolds.
Consider $ \mathcal{M}_n $, the manifold of unordered ``projective'' orthonormal bases of $ \mathbb{C}^n $ (where ``projective'' means that we do not distinguish between a basis vector and its scalar multiples). This manifold can be constructed as a quotient of the complete flag manifold with respect to the symmetric group, which acts by permuting basis vectors. $ \mathcal{M}_n $ can be endowed with the ``MUBness'' distance, thus making it into a metric space; and MUBs correspond to maximally-distant points. For a set of points $q_1, \ldots, q_k \in \mathcal{M}_n $, let $ \mathcal{N} \left( q_1, \ldots, q_k \right) \subseteq \mathcal{M}_n $ denote the subset of points which are unbiased with respect to the $k $ points $ q_1, \ldots, q_k $.
Using this language, we frame a general heuristic procedure for constructing a list of MUBs. This is done iteratively: we first choose an arbitrary $ q_1 \in \mathcal{M}_n $. Then, in the $k$th iteration we choose an arbitrary point $q_k \in \mathcal{N} \left( q_1, \ldots, q_{k-1} \right) $. Any set of MUBs can be constructed in this manner, but the choices in each step are highly redundant. 

Our main contribution is alleviating this redundancy issue. The unitary group $ \mathrm{U} \left( n \right) $ acts on $ \mathcal{M}_n $, since the latter is a quotient of $ \mathrm{U} \left( n \right) $. We show that $ \mathcal{N} \left( q_1, \ldots, q_{k-1} \right) $ is closed under the action of the subgroup $ \mathrm{U} \left( n \right)_{q_1, \ldots, q_{k-1}} \subseteq \mathrm{U} \left( n \right) $ that stabilizes $ q_1, \ldots, q_{k-1} $. We then prove the following (\Cref{thm:main}): two candidates $p, q \in \mathcal{N} \left( q_1, \ldots, q_{k-1} \right) $ for $q_k$ yield equivalent MUB lists, if and only if they belong to the same $ \mathrm{U} \left( n \right)_{q_1, \ldots, q_{k-1}} $-orbit.
For example, consider the submanifold $ \mathcal{N} \left( e \right) $ comprising bases that are unbiased with respect to the standard basis $e$. In this case, our partition of $ \mathcal{N} \left( e \right) $ is essentially equivalent to the standard classification of complex Hadamard matrices.
The partition of $ \mathcal{N} \left( q_1, \ldots, q_k \right) $ is a vast generalization of this classification. We then utilize our results to derive hitherto unknown equivalences between MUB triples in dimension $n=4$.

This paper is organized as follows. \Cref{sec:geometry} outlines our geometric constructions. First we construct $ \mathcal{M}_n $, the space of orthonormal bases for $\mathbb{C}^n$, as a discrete quotient of the complete flag manifold. Next, we introduce the MUBness metric (distance) on $ \mathcal{M}_n $, and show that the unitary group acts on $ \mathcal{M}_n $ via isometries. We then use the MUBness metric to define the subsets $ \mathcal{V} \left( q \right) \subseteq \mathcal{M}_n $ of $q$-unbiased bases; $\mathcal{N} \left( q_1, \ldots, q_k \right)$ is later defined as the intersection $ \bigcap_{i=1}^k \mathcal{V} \left( q_i \right) $. \Cref{sec:procedure} then relates MUB lists to theses geometric structures. We describe a standard procedure for constructing lists of MUBs, using our language. Later we define equivalence of MUB lists, and then we formulate and prove our main result, \Cref{thm:main}. For any step in the aforementioned procedure, this theorem characterizes choices that yield equivalent MUB steps, in terms of group orbits in $ \mathcal{N} \left( q_1, \ldots, q_k \right) $. We then prove that these orbits are discrete. \Cref{sec:app} demonstrates our results in dimension $n=4$. This section serves two purposes: to illustrate concretely the abstract framework of the preceding sections; and to derive new results regarding equivalence of MUB triples, thus showcasing the power of our methods.
We conclude this paper by summarizing our main insights and proposing directions for future research.

\section{The space of orthonormal bases}\label{sec:geometry}
This section details the geometric constructions that shall be used later. Let us briefly outline the upshot of these constructions. First, there exists a set $ \mathcal{M}_n $, whose points correspond to orthonormal bases of $ \mathbb{C}^n $; and $ \mathcal{M}_n $ is endowed with a transitive $ \mathrm{U} \left( n \right) $-action. For each $ q \in \mathcal{M}_n $, there exists a subset $ \mathcal{V} \left( q \right) \subseteq \mathcal{M}_n $ comprising all bases which are unbiased with respect to $q$. $ \mathcal{V} \left( q \right) $ satisfies two key properties with respect to the $ \mathrm{U} \left( n \right) $-action:
\begin{itemize}
	\item For any $ U \in \mathrm{U} \left( n \right) , \; \mathcal{V} \left( U \cdot q \right) = U \cdot \mathcal{V} \left( q \right) $;
	\item $\mathcal{V} \left( q \right)$ is closed under the action of the stabilizer subgroup $ \mathrm{U} \left( n \right)_q $ of $q$.
\end{itemize}
Here $ \cdot $ denotes the $ \mathrm{U} \left( n \right) $-action, and $ \mathrm{U} \left( n \right)_q \defeq \left\{ U \in \mathrm{U} \left( n \right) \mid U \cdot q = q \right\} $. The remainder of this section may be skipped by readers content to accept these results without proof.

\subsection{Constructing $ \mathcal{M}_n $}\label{sec:M_n}
Recall that $\mathcal{M}_n$ denotes the set of all orthonormal bases (ONBs) of $ \mathbb{C}^n $. Here we construct $\mathcal{M}_n$ as a quotient of the unitary group $\mathrm{U}(n)$. $\mathcal{M}_n$ thus obtains the structure of a smooth manifold.

Before proceeding, we first refine our definitions. Recall that global phases are insignificant in quantum mechanics: the vectors $ \ket{\psi} $ and $ e^{i \theta} \ket{\psi} $ correspond to the same physical state. If a quantum system is described by the Hilbert space $\mathcal{H} $, then its possible states are the \textit{rays} in $\mathcal{H} $, i.e. the elements of the projective space $ \mathbb{P} \left( \mathcal{H} \right) $. Moreover, from a mathematical perspective, the definition \eqref{MUB_def} is invariant under multiplication of any basis vector by a phase. Hence, hereon we refer to basis elements ``in the projective sense''; i.e., we think of basis elements as rank-one projections (equivalently: points in $\mathbb{CP}^{n-1}$) rather than actual vectors. We say that a set $ p = \left\{ P_i \right\}_{i=1}^n $ of pairwise-orthogonal rank-one projections is \textit{complete} if $ \sum_{i=1}^n P_i = I $. Hereon, ``orthonormal basis'' is taken to mean ``complete set of pairwise-orthogonal rank-one projections''.
Two such complete sets $ p $ and $ q = \left\{ Q_j \right\}_{j=1}^n $ are said to be unbiased if:
\begin{equation}
	\forall i,j, \quad \Tr \left( P_i Q_j \right) = \frac{1}{n} .
\end{equation}

Any unitary matrix $ U \in \mathrm{U}(n) $ defines an orthonormal basis by taking its columns as the basis vectors. Thus we obtain a map $ \pi : \mathrm{U}(n) \twoheadrightarrow \mathcal{M}_n $. Since any basis can be obtained from some unitary matrix, this map is surjective; however it is not injective. In fact, there are two obstructions to injectivity:
\begin{itemize}
	\item We do not care about multiplication of basis vectors by phases;
	\item We do not care about the order of the basis vectors.
\end{itemize}
These obstructions are symmetries of the fibers $ \pi^{-1} \left( q \right) $; in other words, there are two group actions on $\mathrm{U}(n) $ that preserve the fibers. 
The first action multiplies each column by some phase $ \mathrm{U}(1) $. 
Clearly, this action is given by right multiplication with an element of the torus subgroup $ \mathrm{T}^n $ (the diagonal unitaries). Thus, the quotient of $ \mathrm{U}(n) $ by the first action is precisely the \textit{complete flag manifold}:
\begin{equation}
	\tilde{\mathcal{M}}_n \defeq \mathrm{Flag} \left( 1, 2, \ldots, n \right) = \mathrm{SL} \left( n; \mathbb{C} \right) / \mathrm{B}^n = \mathrm{SU} \left( n \right) / \left( \mathrm{T}^n \cap \mathrm{SU} \left( n \right) \right) = \mathrm{U} \left( n \right) / \mathrm{T}^n ,
\end{equation}
where $ \mathrm{B}^n $ is a Borel subgroup of $ \mathrm{SL} \left( n; \mathbb{C} \right) $.

Note we have only resolved the first obstruction, since flag manifolds are sensitive to the ordering of basis vectors. Orthonormal bases correspond to complete flags via:
\begin{equation}
	\left\{ \mathbf{b}_i \right\}_{i=1}^n\longmapsto (V_1,\dots,V_n)
	, \text{ where } V_i \defeq \mathrm{span}_{\mathbb{C}} \left\{ \mathbf{b}_1, \ldots, \mathbf{b}_i \right\} ,
\end{equation}
so any reordering of the basis elements $ \left\{ \mathbf{b}_i \right\}_{i=1}^n $ results in a different flag.
To eliminate this redundancy we take a second quotient, this time by the action of the symmetric group $ \mathrm{S}_n $, that acts by permuting the order of basis vectors. Let us consider $ \mathrm{S}_n $ as a subgroup of $ \mathrm{U} \left( n \right) $, embedded as the permutation matrices; then reordering of basis elements corresponds to right multiplication by a permutation matrix.

Of course, we can also construct our space with a single quotient rather than take two consecutive ones. 
Define a third subgroup of $ \mathrm{U} \left( n \right) $ -- the semidirect product $ \mathrm{C}_n \defeq \mathrm{S}_n \ltimes \mathrm{T}^n $. This subgroup consists of complex matrices with exactly one nonzero entry in every row and every column, and the nonzero entries all have modulus $1$. Notably, we have $ \mathrm{C}_n = N \left( \mathrm{T}^n \right) $, i.e. this group is the \textit{normalizer subgroup} of the torus in $ \mathrm{U} \left( n \right) $; and the quotient $ N \left( \mathrm{T}^n \right) / \mathrm{T}^n \cong \mathrm{S}_n $ is the \textit{Weyl group} of $ \mathrm{U} \left( n \right) $ (and also of $ \mathrm{SU} \left( n \right) $).
$\mathrm{C}_n$ is a closed subgroup of $ \mathrm{U} \left( n \right) $ (closedness follows because $\mathrm{T}^n$ is a torus and $\mathrm{S}_n$ is finite), and it has $n!$ connected components.
We define $\mathcal{M}_n$ as the quotient of $ \mathrm{U} \left( n \right) $ by the right action of $\mathrm{C}_n$:
\begin{equation}
	\mathcal{M}_n \defeq \mathrm{U} \left( n \right) / \mathrm{C}_n .
\end{equation}
This space differs from $ \tilde{\mathcal{M}}_n $ only by the quotient of a discrete group, so they both have the same dimension and share many other important features. Crucially, both are compact connected homogeneous $ \mathrm{U} \left( n \right) $-spaces. Moreover, $ \tilde{\mathcal{M}}_n $ is an $n!$-fold covering space of $ \mathcal{M}_n $; the covering map is the natural projection $ \tilde{\mathcal{M}}_n \twoheadrightarrow \mathcal{M}_n $ that maps an ordered orthonormal basis to an unordered one. In fact, since the flag manifold $ \tilde{\mathcal{M}}_n $ is simply connected it is the universal cover of $ \mathcal{M}_n $, hence $ \pi_1 \left( \mathcal{M}_n \right) \cong \mathrm{S}_n $.
The space $ \mathcal{M}_n $ is the \textit{permutation invariant complete flag manifold}, or the \textit{symmetrized complete flag manifold}, for $ \mathbb{C}^n $. It was defined in the same way in \cite{tirkkonen2017grassmannian}.
We note that the symmetrized space was less studied in the literature and is often not as convenient to work with as $ \tilde{\mathcal{M}}_n $.

Hereon, the points of $ \mathcal{M}_n $ are denoted either by small Latin letters (usually $p$ or $q$, with or without a subscript), or in the form $ U \mathrm{C}_n $ (the coset of $ \mathrm{C}_n $ represented by $ U \in \mathrm{U} \left( n \right) $).

\subsection{MUBness}\label{sec:MUBness}
In \Cref{sec:procedure}, we use $\mathcal{M}_n$ to define a procedure that yields mutually unbiased bases. To do so, we need to enrich $\mathcal{M}_n$ with additional geometric structure that captures the notion of mutual-unbiasedness, or \textit{MUBness}.
Given any two ONBs $p = \left\{ P_i \right\}_{i=1}^n $ and $q= \left\{ Q_j \right\}_{j=1}^n$, their MUBness is defined as~\cite{bengtsson2007three}:
\begin{equation}\label{MUBness}
	D \left( p, q \right) \defeq \sqrt{ n-1- \sum_i \sum_j \left( \Tr \left( P_i Q_j \right) -\frac{1}{n} \right)^2} .
\end{equation}
Evidently, $ 0 \leq D \left( p, q \right) \leq \sqrt{ n-1} $; and $ D \left( p, q \right) = \sqrt{ n-1} $ iff $p,q$ are mutually unbiased. Therefore, $D \left( p, q \right)$ indeed captures the extent of ``MUBness'' between $p$ and $q$. Here we recast the original derivation of \eqref{MUBness} in our own terms; namely, we construct a smooth embedding $\phi$ of $ \mathcal{M}_n $ into a certain Grassmannian manifold.
The pullback of the chordal metric (distance function) on the Grassmannian by $\phi$ then defines the MUBness metric on $ \mathcal{M}_n $. In fact, we shall see that $\phi$ embeds $\mathcal{M}_n$ as a single orbit of a $\mathrm{U} \left( n \right)$-action on the Grassmannian. Since $\mathrm{U} \left( n \right)$ acts on the Grassmannian by isometries, we get that the $\mathrm{U} \left( n \right)$ action on $ \mathcal{M}_n $ respects the MUBness metric.

First, recall the definition of $\mathcal{M}_n$ as the coset space $ \mathrm{U} \left( n \right) / \mathrm{C}_n $, and denote its points as cosets $ V \mathrm{C}_n $ for $ V \in \mathrm{U} \left( n \right) $. Then $ \mathrm{U} \left( n \right) $ acts on $\mathcal{M}_n$ transitively by left multiplication:
\begin{equation}\label{Un_action}
	\forall U \in \mathrm{U} \left( n \right) , \quad U \cdot V \mathrm{C}_n \defeq \left( UV \right) \mathrm{C}_n .
\end{equation}
Each $ V \mathrm{C}_n \in \mathcal{M}_n $ can be represented as the complete set of orthogonal projections $ p = \left\{ P_i \right\}_{i=1}^n $, where $ P_i = \ket{v_i} \bra{v_i} $ for $ \ket{v_i} $ the $i$th column of $V$.

Second, we define the Grassmannian and its $ \mathrm{U} \left( n \right) $ action. Using the physicists' convention, $ \mathfrak{su} \left( n \right) $ is the Lie algebra of traceless Hermitian matrices. We consider $ \mathfrak{su} \left( n \right) $ as a real vector space with the Frobenius inner product (which equals the Killing form up to a scalar multiple). Let $ \mathrm{Gr} \left( n-1, \mathfrak{su} \left( n \right) \right) $ be the Grassmannian of $ \left( n-1 \right) $-dimensional subspaces of $ \mathfrak{su} \left( n \right) $. This Grassmannian is naturally endowed with an $ \mathrm{SO} \left( \mathfrak{su} \left( n \right) \right) $-action: for $O \in \mathrm{SO} \left( \mathfrak{su} \left( n \right) \right) $ and $W \in \mathrm{Gr} \left( n-1, \mathfrak{su} \left( n \right) \right) $, $ O \cdot W $ is the image $ O \left( W \right) $. 
Now, $ \mathrm{U} \left( n \right) $ acts on $ \mathfrak{su} \left( n \right) $ via $ \mathrm{Ad}_U \left( A \right) \defeq U A U^\dagger $. This action preserves the Frobenius inner product, i.e. defines an element of $ \mathrm{O} \left( \mathfrak{su} \left( n \right) \right) $; and since $ \mathrm{U} \left( n \right) $ is connected, its image must land in the connected component of the identity, i.e. $ \mathrm{SO} \left( \mathfrak{su} \left( n \right) \right) $. Thus, we get a $ \mathrm{U} \left( n \right) $-action on $ \mathrm{Gr} \left( n-1, \mathfrak{su} \left( n \right) \right) $.
We define the chordal Grassmannian distance~\cite{bengtsson2007three} by $ d_c \left( W_1, W_2 \right) \defeq \frac{1}{2} \mathrm{Tr} \left( S_{W_1} -S_{W_2} \right)^2 $, where $ S_W $ denotes the unique orthogonal projection onto the subspace $ W \in \mathrm{Gr} \left( n-1, \mathfrak{su} \left( n \right) \right) $. This metric is $ \mathrm{SO} \left( \mathfrak{su} \left( n \right) \right) $-invariant, since:
\begin{multline*}
	\forall O \in \mathrm{SO} \left( \mathfrak{su} \left( n \right) \right), \; d_c \left( O \cdot W_1, O \cdot W_2 \right) = \frac{1}{2} \mathrm{Tr} \left( O S_{W_1} O^T -O S_{W_2} O^T \right)^2 = \\
	= \frac{1}{2} \mathrm{Tr} \left[ O \left( S_{W_1}- S_{W_2} \right) O^T \right]^2 = \frac{1}{2} \mathrm{Tr} \left( S_{W_1} -S_{W_2} \right)^2 = d_c \left( W_1, W_2 \right) .
\end{multline*}

We now arrive at the definition of the map $\phi$.
Given an ONB $p = \left\{ P_i \right\}_{i=1}^n \in \mathcal{M}_n $ (recall each $P_i$ is a rank-one projection), we define:
\begin{equation}\label{Pr_to_Gr}
	\phi \left( p \right) \defeq \mathrm{span}_{\mathbb{R}} \left\{ P_i -\frac{\mathds{1}}{n} \right\}_{i=1}^n .
\end{equation}
\begin{proposition}\label{prop:phi}
	\eqref{Pr_to_Gr} defines a smooth embedding $ \phi : \mathcal{M}_n \hookrightarrow \mathrm{Gr} \left(  n-1, \mathfrak{su} \left( n \right) \right) $. Moreover, $\phi$ is equivariant with respect to the $ \mathrm{U} \left( n \right) $-actions on both spaces.
\end{proposition}
\begin{proof}
	By definition of $ \mathfrak{su} \left( n \right) $, we have $ P_i -\frac{\mathds{1}}{n} \in \mathfrak{su} \left( n \right) $ for all $i$ (more generally, the set of all density matrices acting on $ \mathbb{C}^n $ is embedded in the affine space $ \mathfrak{su} \left( n \right) + \frac{\mathds{1}}{n} $ of unit-trace Hermitian operators). To see that the RHS of \eqref{Pr_to_Gr} has dimension $n-1$, first note that the elements $ P_i -\frac{\mathds{1}}{n} $ sum up to zero, hence are not linearly independent. Moreover, we now show that the subset (say) $ \left\{ P_i -\frac{\mathds{1}}{n} \right\}_{i=1}^{n-1} $ \emph{is} linearly independent. If a real linear combination $ \sum_{i=1}^{n-1} a_i \left( P_i -\frac{\mathds{1}}{n} \right) $ vanishes, then
	\begin{equation*}
		\sum_{i=1}^{n-1} a_i P_i = \frac{1}{n} \sum_{i=1}^{n-1} a_i \mathds{1} .
	\end{equation*}
	Now, consider the matrix ranks of both sides. Since the $P_i$ are all rank-one projections and the pairwise intersections $ \mathrm{im} \left( P_i \right) \cap \mathrm{im} \left( P_j \right) $ are all trivial (for $i \neq j$), we deduce that $ \mathrm{rank} \left( \sum_{i=1}^{n-1} a_i P_i \right) = \# \left\{ 1 \leq i \leq n-1 \mid a_i \neq 0 \right\} $. In contrast, the rank of the RHS is zero if $ \sum_{i=1}^{n-1} a_i = 0 $ and $n$ otherwise. The only way for the two ranks to agree is if $a_i=0$ for all $i$, as required. Thus, $ \left\{ P_i -\frac{\mathds{1}}{n} \right\}_{i=1}^n $ indeed span an $\left( n-1 \right)$-dimensional subspace of $ \mathfrak{su} \left( n \right) $; and the definition \eqref{Pr_to_Gr} is independent of the order of the basis elements $P_i$; thus, $\phi$ is indeed a well-defined map $ \mathcal{M}_n \rightarrow \mathrm{Gr} \left(  n-1, \mathfrak{su} \left( n \right) \right) $.
	
	Next, we show that $\phi$ is equivariant. For any $ U \in \mathrm{U} \left( n \right) $ and $ p \in \mathcal{M}_n $, we have: 
	\begin{multline*}
		\phi \left( U \cdot p \right) = \mathrm{span}_{\mathbb{R}} \left\{ U P_i U^\dagger -\frac{\mathds{1}}{n} \right\}_{i=1}^n = \mathrm{span}_{\mathbb{R}} \left\{ U \left( P_i -\frac{\mathds{1}}{n} \right) U^\dagger \right\}_{i=1}^n = \\ 
		= \mathrm{span}_{\mathbb{R}} \left\{ \mathrm{Ad}_U \left( P_i -\frac{\mathds{1}}{n} \right) \right\}_{i=1}^n = \mathrm{Ad}_U \left( \mathrm{span}_{\mathbb{R}} \left\{ P_i -\frac{\mathds{1}}{n} \right\}_{i=1}^n \right) = \mathrm{Ad}_U \left( \phi \left( p \right) \right) ,
	\end{multline*}
	where $ \mathrm{span}_{\mathbb{R}} $ and $ \mathrm{Ad}_U $ commute since $ \mathrm{Ad}_U $ is an invertible $\mathbb{R}$-linear transformation.
	
	We can now show that $\phi$ is smooth, using transitivity of the $ \mathrm{U} \left( n \right) $-action on $ \mathcal{M}_n $ and equivariance of $\phi$. First, note $\phi$ is determined by its value on a single point, say the standard basis $ \mathds{1} \mathrm{C}_n \in \mathcal{M}_n $ (the coset of the identity matrix $\mathds{1}$). This is true because $ \phi \left( U \mathrm{C}_n \right) = \phi \left( U \cdot \mathds{1} \mathrm{C}_n \right) = \mathrm{Ad}_U \left( \phi \left( \mathds{1} \mathrm{C}_n \right) \right) $.
	Now, the map $ \pi : \mathrm{U} \left( n \right) \twoheadrightarrow \mathcal{M}_n $ is a surjective smooth submersion, hence by the characteristic property of surjective smooth submersions (Theorem 4.29 in \cite{Lee2012Smooth_Ch4}), $\phi$ is smooth if and only if $ \phi \circ \pi : \mathrm{U} \left( n \right) \rightarrow \mathrm{Gr} \left(  n-1, \mathfrak{su} \left( n \right) \right) $ is smooth.
	For any $ U \in \mathrm{U} \left( n \right) $ we have $ \phi \circ \pi \left( U \right) = \phi \left( U \mathrm{C}_n \right) = \mathrm{Ad}_U \left( \phi \left( \mathds{1} \mathrm{C}_n \right) \right) $, which is smooth since the $ \mathrm{SO} \left( \mathfrak{su} \left( n \right) \right) $-action on the Grassmannian is smooth.
	
	The equivariant rank theorem implies that $\phi$ has constant rank. Thus, if $\phi$ is injective then it is a smooth immersion, hence also a smooth embedding (since its domain $\mathcal{M}_n$ is compact). Let $p,q \in \mathcal{M}_n$, and suppose $ \phi \left( p \right) = \phi \left( q \right) $. By transitivity, there exists $ U \in \mathrm{U} \left( n \right) $ such that $ q = U \cdot p $; hence $  \phi \left( p \right) =  \phi \left( U \cdot p \right) = \mathrm{Ad}_U \left( \phi \left( p \right) \right) $, i.e. $ U \in \mathrm{Stab}_{\phi \left( p \right)} $, the stabilizer of $\phi \left( p \right) \in \mathrm{Gr} \left(  n-1, \mathfrak{su} \left( n \right) \right)$. Therefore, $\phi$ is injective if and only if $ U \in \mathrm{Stab}_{\phi \left( p \right)} $ implies $ U \cdot p = p $, i.e. $ \mathrm{Stab}_{\phi \left( p \right)} $ is a subgroup of the stabilizer of $p \in \mathcal{M}_n$, which is isomorphic to $\mathrm{C}_n$. Equivalently, the stabilizers are isomorphic, since $ \mathrm{Stab}_p \subseteq \mathrm{Stab}_{\phi \left( p \right)} $ follows directly from equivariance.
	By transitivity of the $ \mathrm{U} \left( n \right) $-action on $ \mathcal{M}_n $ and equivariance of $\phi$, the image $ \phi \left( \mathcal{M}_n \right) $ comprises a single orbit, hence a homogeneous $ \mathrm{U} \left( n \right) $-space. Thus, the stabilizer subgroups for all points of $\phi \left( \mathcal{M}_n \right)$ are isomorphic; hence, it suffices to show that $ \mathrm{Stab}_{\phi \left( \mathds{1} \mathrm{C}_n \right)} \cong \mathrm{C}_n $.
	
	The standard basis $ \mathds{1} \mathrm{C}_n $ is represented by the projections $\left\{ \ket{i} \bra{i} \right\}_{i=1}^n $, which span the subalgebra $ \mathfrak{t} \subseteq \mathfrak{su} \left( n \right) $ of traceless diagonal Hermitian matrices. A unitary $ U \in \mathrm{U} \left( n \right) $ belongs to the stabilizer of $ \mathfrak{t} \in \mathrm{Gr} \left(  n-1, \mathfrak{su} \left( n \right) \right) $ if and only if $ U \mathfrak{t} U^\dagger = \mathfrak{t} $. Consider a matrix $A \in \mathfrak{t}$ with distinct diagonal entries. $ U A U^\dagger $ is diagonal if and only if the columns of $U^\dagger$ form an orthonormal eigenbasis of $A$. Since the entries of $A$ are distinct, it has a unique orthonormal eigenbasis -- the standard basis -- up to phases and reordering. This precisely means that $ U^\dagger \in \mathrm{C}_n $, which holds if and only if $ U \in \mathrm{C}_n $. This completes the proof.
\end{proof}

This proposition has several corollaries. First, since $\phi$ is injective, we can use it to pull back $ d_c $ and obtain a metric on $ \mathcal{M}_n $. As already established in \cite{bengtsson2007three}, the pullback metric is the MUBness from \eqref{MUBness}. As mentioned above, these facts imply that $ \mathrm{U} \left( n \right) $ acts on $ \mathcal{M}_n $ by isometries with respect to the MUBness metric -- a fact we shall use profusely in the next subsection.

Moreover, the Grassmannian $ \mathrm{Gr} \left( n-1, \mathfrak{su} \left( n \right) \right) $ is equipped with a tautological (real) vector bundle of rank $n-1$: the fiber over $ V \in \mathrm{Gr} \left( n-1, \mathfrak{su} \left( n \right) \right) $ is the vector space $V$. The pullback of this bundle defines a vector bundle $ \mathcal{B} \rightarrow \mathcal{M}_n $, where the fiber over $ p = \left\{ P_i \right\}_{i=1}^n \in \mathcal{M}_n $ is given by $ \mathrm{span}_{\mathbb{R}} \left\{ P_i -\frac{\mathds{1}}{n} \right\}_{i=1}^n $. Since the tautological bundle on the Grassmannian is $ \mathrm{SO} \left( \mathfrak{su} \left( n \right) \right) $-equivariant, $ \mathcal{B} $ is $ \mathrm{U} \left( n \right) $-equivariant: for a pair $ \left( p, A \right) $ with $ p \in \mathcal{M}_n $ and $ A \in \phi \left( p \right) $, define $ U \cdot \left( p, A \right) \defeq \left( U \cdot p, U A U^\dagger \right) $. The data of the equivariant embedding $\phi : \mathcal{M}_n \hookrightarrow \mathrm{Gr} \left( n-1, \mathfrak{su} \left( n \right) \right) $ is the same as that of the equivariant rank-$\left( n-1 \right)$ vector bundle $ \mathcal{B} $. Thus, we get another perspective as to why $ \mathcal{M}_n $ embeds in the Grassmannian, hence a new way to look at the geometric origins of MUBness.

As an aside, we note that the space of global sections $ \Gamma \left( \mathcal{B} \right) $ is naturally a representation of $ \mathrm{U} \left( n \right) $.
The embedding $\phi$ is somewhat analogous to how the flag manifold $ \tilde{\mathcal{M}}_n $ (a homogeneous $ \mathrm{SL}_n$-space) embeds as a closed $ \mathrm{SL}_n \left( \mathbb{C} \right) $-orbit in $ \mathbb{P} V $, where $V$ corresponds to an irreducible representation of $ \mathrm{SL}_n \left( \mathbb{C} \right) $. The pullback of the tautological line bundle on $ \mathbb{P} V $ yields an equivariant line bundle, whose space of global sections is a representation of $ \mathrm{SL}_n \left( \mathbb{C} \right) $, equivalent to $V$. Alternatively, $\phi$ can be considered analogous to the embedding of $ \tilde{\mathcal{M}}_n $ as a coadjoint orbit in $ \mathfrak{sl} \left( n; \mathbb{C} \right)^* $.

\subsection{Subsets of $q$-unbiased bases}\label{sec:Vq}
Let $ p, q \in \mathcal{M}_n $ be any two ONBs; recall that $ D^2 \left( p, q \right) = n-1$ iff $p$ and $q$ are unbiased. Therefore:
\begin{equation}
	\mathcal{V} \left( q \right) \defeq \left\{ p \in \mathcal{M}_n \mid D^2 \left( p, q \right) = n-1 \right\} 
\end{equation}
defines $ \mathcal{V}\left( q \right) \subseteq \mathcal{M}_n $ as the subset of all bases that are $ q $-unbiased (i.e. unbiased with respect to $q$).
For example, consider the standard basis $e \defeq \mathds{1} \mathrm{C}_n $ (the coset of the identity matrix $ \mathds{1} $). We have 
\begin{equation}
	\mathcal{V} \left( e \right) = \left\{ H \mathrm{C}_n \mid H \in \mathrm{U} \left( n \right) \text{ and } \forall i,j, \, \abs{H_{ij}}=1/\sqrt{n} \right\} .
\end{equation}
Matrices $H$ that satisfy these conditions (unitary matrices where all entries have modulus $ 1 / \sqrt{n} $) are called \textit{complex Hadamard matrices}.
Two Hadamard matrices $ H_1, H_2 $ represent the same point in $ \mathcal{V} \left( e \right) $ iff $ H_1^\dagger H_2 \in \mathrm{C}_n $. Note that $ \mathcal{V} \left( q \right) $ is a closed subset of $ \mathcal{M}_n $ (for every $q$), but generally it may not be a submanifold (i.e. it may be singular). Moreover, any two $ \mathcal{V} \left( q \right) $ are homeomorphic. Indeed, let $ U \in \mathrm{U} \left( n \right) $ be a unitary that maps $ q $ to $ q' $ via the left action: $ q' = U \cdot q $ (transitivity of the action means that such a unitary always exists). Then we have:
\begin{align}\label{V_q_diffeo}
	\mathcal{V} \left( U \cdot q \right) & = \left\{ p \in \mathcal{M}_n \mid D^2 \left( p, U \cdot q \right) = n-1 \right\} = \nonumber\\
	& = \left\{ p \in \mathcal{M}_n \mid D^2 \left( U^\dagger \cdot p, q \right) = n-1 \right\} = \left\{ p \in \mathcal{M}_n \mid U^\dagger \cdot p \in \mathcal{V} \left( q \right) \right\} = \nonumber\\
	& = U \cdot \mathcal{V} \left( q \right) ,
\end{align}
using the fact that $ \mathrm{U} \left( n \right) $ acts via isometries. 
Let $ \mathrm{U} \left( n \right)_q \defeq \left\{ U \in \mathrm{U} \left( n \right) \mid U \cdot q = q \right\} $ denote the isotropy (stabilizer) subgroup. By the definition of a quotient space, if $q = V \mathrm{C}_n$ then $ \mathrm{U} \left( n \right)_q = V \mathrm{C}_n V^\dagger $.
Now, let $ U \in \mathrm{U} \left( n \right)_q $; from \eqref{V_q_diffeo} we see that $ U $ defines an \textit{automorphism} $ U : \mathcal{V} \left( q \right) \rightarrow \mathcal{V} \left( q \right) $. Thus, for $ q = V \mathrm{C}_n $ we have that $ \mathcal{V} \left( V \mathrm{C}_n \right) $ is closed under the action of the isotropy group $ V \mathrm{C}_n V^\dagger $. For example, the set $ \mathcal{V} \left( e \right) $ of complex Hadamard matrices is closed under left multiplication by a $ \mathrm{C}_n $ matrix.

\section{Mutually unbiased bases in the space of orthonormal bases}\label{sec:procedure}
In this section, we consider lists (finite ordered sets) of MUBs. Every MUB list may be constructed via a generic iterative procedure. Let us describe this procedure informally. We start with an empty list, and each iteration adds one additional basis. The additional basis is chosen out of the set $X$, which is guaranteed to comprise all viable choices for the next basis. We initialize $X$ as $\mathcal{M}_n$; in the $i$th iteration, we choose some $q \in X$; and then we replace $X$ by its subset of $q$-unbiased bases. We stop under one of two conditions: either we reached the desired number of MUBs; or $X$ is empty. The latter means we constructed an \textit{unextendible} MUB list, i.e., there is no basis which is mutually unbiased with respect to every element of the list.


Using the definitions from the previous section, we can see that the $q$-unbiased points in $X$ are given by $ X \cap \mathcal{V} \left( q \right) $. Therefore, after $k$ iterations, the set $X$ is given by:
\begin{equation}
	\mathcal{N} \left( q_1, \ldots, q_k \right) \defeq \bigcap_{j=1}^k \mathcal{V} \left( q_j \right) .
\end{equation}
The elements of $\mathcal{N} \left( q_1, \ldots, q_k \right)$ comprise the ONBs which are unbiased with respect to $q_i$ for all $i = 1, 2, \ldots, k$.
However, explicit descriptions of the subsets $ \mathcal{N} \left( q_1, \ldots, q_k \right) $ may be difficult to obtain; the lack of such explicit descriptions means that our procedure is conceptual rather than computational. 

We now describe our procedure formally:
\begin{algorithmic}[1]
	\Require $n \geq 1$, $ m \in \left\{ 1, \ldots, n+1 \right\} $
	\Ensure $\forall i \neq j, \; q_i, q_j \in \mathcal{M}_n$ are mutually unbiased, and $ \left\{ q_i \right\}_{i=1}^{k-1} $ is unextendible or has length $m$
	\State $ k \Leftarrow 1 $
	\While{$\mathcal{N} \left( q_1, \ldots, q_{k-1} \right) \neq \emptyset$ and $ k \leq m $}
	\State \textbf{choose} $ q_{k} \in \mathcal{N} \left( q_1, \ldots, q_{k-1} \right) $
	\State $ k \Leftarrow k+1 $
	\EndWhile
\end{algorithmic}
For $k=1$, $ \mathcal{N} \left( \varepsilon \right) $ is interpreted as $ \mathcal{M}_n $, where $ \varepsilon $ denotes the ``empty list'' of $0$ points.

We now define a notion of equivalence between lists of MUBs.
\begin{definition}\label[definition]{def:equiv_MUB_lists}
	Let $ \left( p_1, \ldots, p_k \right) $ and $ \left( q_1, \ldots, q_k \right) $ be two lists of $k$ MUBs in dimension $n$. These two MUB lists are said to be \emph{equivalent} if there exists a unitary $U \in \mathrm{U} \left( n \right)$ such that $ U \cdot p_i = q_i $ for all $i$.
\end{definition}
The remainder of this section deals with the following question: when does the above procedure produce equivalent MUB lists? As we shall see, the subset $\mathcal{N} \left( q_1, \ldots, q_{k-1} \right)$ is closed under the action of a certain $ \mathrm{U} \left( n \right) $-subgroup: the \textit{simultaneous stabilizer} of $ q_1, \ldots, q_{k-1} $. And two choices for $q_k$ in the $k$-th step of the procedure yield equivalent MUB lists, if and only if they lie in the same orbit of the simultaneous stabilizer. This is the content of \Cref{thm:main}. Next, we study properties of these orbits, and show they are finite and discrete for any $k>2$.

Note we refer to $\mathcal{N} \left( q_1, \ldots, q_k \right)$ as mere subsets; in fact, their partition into orbits (of the simultaneous stabilizer) corresponds to a \textit{stratification}. However, we do not use this fact.

\subsection{Resolving redundancies}\label{sec:redundancies}
In the beginning of the current section we outlined a procedure that can yield any MUB list. Let us take a closer look at this procedure, and review what is already known about equivalent MUB lists.

We start by noting that existing analytic constructions of MUBs have utilized procedures similar to ours. Brierley and Weigert~\cite{brierley2009all} construct all MUB sets in dimensions $2$ to $5$. Given a list of $k-1$ MUBs, they compute all vectors which are unbiased to all elements of the current bases, and then find all the ways to put them together into ONBs.
To simplify the enumeration of all MUB sets, they define a \textit{standard form} for MUB lists, and claim that any MUB list is equivalent to a standard-form one. The following Lemma describes the standard form using this paper's terminology.
\begin{lemma}[\cite{brierley2009all}]\label[lemma]{lemma:MUB_list_standard_form}
	Any list of $k$ MUBs is equivalent to a list of the form $ \left( \mathds{1} \mathrm{C}_n, H_2 \mathrm{C}_n \ldots, H_k \mathrm{C}_n \right) $ that obeys the following conditions:
\begin{enumerate}[label=(\roman*)]
	\item The first basis is $e = \mathds{1} \mathrm{C}_n$, the standard basis;
	\label{cond_I}
	\item The matrices $H_2 \ldots, H_k$ are complex Hadamard matrices;
	\label{cond_H}
	\item The entries in the first column of $H_2$ all equal $1 / \sqrt{n}$;
	\label{cond_col}
	\item The first row of each of the Hadamard matrices $H_2 \ldots, H_k$ has only the entries $1 / \sqrt{n}$.
	\label{cond_row}
\end{enumerate}
\end{lemma}
\begin{proof}
	Let $ \left( q_1, \ldots, q_k \right) $ be a MUB list. Since $ \mathcal{M}_n $ is a homogeneous space, there always exists a unitary $U$ such that $U \cdot q_1 = e$. Thus, as a first step we choose such $U$ and note that $ \left( q_1, \ldots, q_k \right) $ is equivalent to $ \left( U \cdot q_1, \ldots, U \cdot q_k \right) $. Now, $ U \cdot q_1 = e $; and since it is a MUB list, we must have $ U \cdot q_i \in \mathcal{N} \left( e \right) $ for all $i>1$. As we had mentioned previously, $ \mathcal{N} \left( e \right) = \mathcal{V} \left( e \right) $ comprises all bases $H \mathrm{C}_n$ where $H$ is a $n \times n$ complex Hadamard matrix. Thus, any MUB list is equivalent to one of the form $ \left( \mathds{1} \mathrm{C}_n, H_2 \mathrm{C}_n \ldots, H_k \mathrm{C}_n \right) $, where $H_i$ are complex Hadamard matrices. This fact ensures that \cref{cond_I,cond_H} can always be satisfied.
	
	Next, recall that all entries of a Hadamard matrix have the form $e^{i \theta}/\sqrt{n}$; so we can ``dephase'' the first row of each $H_i$ by replacing $H_i \rightarrow H_i D$, where $D \in \mathrm{C}_n$ is a suitable diagonal matrix. But how can we satisfy the third condition? Similar to how rows can be dephased via right-multiplication by a $\mathrm{C}_n$ element, a column can be dephased via \textit{left}-multiplication by a $\mathrm{C}_n$ element. Explicitly: let $ \left( \mathds{1} \mathrm{C}_n, H_2 \mathrm{C}_n, \ldots, H_k \mathrm{C}_n \right) $ be a list that obeys \cref{cond_I,cond_H,cond_row}, and let the $ \left( 1, e^{i \theta_2}, \ldots, e^{i \theta_n} \right)^T / \sqrt{n} $ be the first column of $H_2$. Define $D$ to be the diagonal matrix with entries $ e^{-i \theta_j} $; clearly $D H_2$ satisfies \cref{cond_col}. Of course, we do not have the freedom to replace $H_2$ by $ D H_2 $; but we do have the freedom to act on \textit{the entire list} with $D$ on the left. 
	Crucially, the first element remains unchanged, since $ D \cdot \mathds{1} \mathrm{C}_n = D \mathrm{C}_n = \mathds{1} \mathrm{C}_n $ (because $D \in \mathrm{C}_n$); and one easily verifies that each $D H_i$ is a complex Hadamard matrix with the same first row as $H_i$. Thus, the new list $ \left( \mathds{1} \mathrm{C}_n, D H_2 \mathrm{C}_n, \ldots, D H_k \mathrm{C}_n \right) $ satisfies \cref{cond_I,cond_H,cond_col,cond_row}, as required.
\end{proof}
As is evident from the proof, there is a deep reason why \cref{cond_I,cond_col} can be satisfied simultaneously. We can replace $ \left( q_1, \ldots, q_k \right) \rightarrow \left( U \cdot q_1, \ldots, U \cdot q_k \right) $, so we choose $U$ that fixes the first basis to be the standard one. However, this condition does not fix $U$ completely; rather, we can multiply it by an element of $ \mathrm{C}_n $, which is the stabilizer group of $e$. Indeed, if $U \cdot q_1 = \mathds{1} \mathrm{C}_n$ and $P \in \mathrm{C}_n$, then $ PU \cdot q_1 = P \cdot \mathds{1} \mathrm{C}_n = P \mathrm{C}_n = \mathds{1} \mathrm{C}_n $.

We can put this insight in more general terms. First, consider $k=1$, i.e. lists $ \left( q_1 \right) $ comprising a single basis (it is a MUB list, vacuously). By homogeneity of $\mathcal{M}_n$, all such lists are equivalent; so without loss of generality, we may assume $q_1 = \mathds{1} \mathrm{C}_n$. Now, suppose we have two lists of length $k=2$, and we wish to decide whether they are equivalent or not. We can bring the lists to the form $ \left( \mathds{1} \mathrm{C}_n, p \right) $, $ \left( \mathds{1} \mathrm{C}_n, q \right) $, where $p, q \in \mathcal{V} \left( e \right)$.
The two lists are equivalent iff there exists a unitary $U$ that obeys $ U \cdot p = q $ \textit{and} $ U \cdot \mathds{1} \mathrm{C}_n = \mathds{1} \mathrm{C}_n $. The latter condition means that $U \in \mathrm{C}_n$. As we have seen in \Cref{sec:Vq}, $\mathcal{V} \left( e \right)$ is closed under the action of $\mathrm{C}_n$. However, it is not always a homogeneous space. Hence, there may not exist a unitary $U \in \mathrm{C}_n$ that maps $p$ to $q$. By definition, such a unitary exists iff $p, q$ lie in the same orbit of the $\mathrm{C}_n$ action. 
Recall that two complex Hadamard matrices $H_1, H_2$ are said to be equivalent if there exist permutation matrices $ P_1, P_2 $ and diagonal unitary matrices $ D_1, D_2 $ such that:
\begin{equation}
	H_1 = D_1 P_1 H_2 P_2 D_2 .
\end{equation}
Since the torus is normal in $ \mathrm{C}_n $, there exists some $ \tilde{D}_1 \in \mathrm{T}^n $ s.t. $ P_1 \tilde{D}_1 = D_1 P_1 $. Hence we can replace the above condition by:
\begin{equation}
	H_1 = P_1 \tilde{D}_1 H_2 P_2 D_2 ,
\end{equation}
where it is transparent that $H_1, H_2$ are equivalent iff the points $ H_1 \mathrm{C}_n, H_2 \mathrm{C}_n \in \mathcal{V} \left( e \right) $ lie in the same $\mathrm{C}_n$-orbit. Therefore, the classification of complex Hadamard matrices is equivalent to the classification of MUB pairs.

We now state and prove a generalization of this idea. First, recall from \Cref{sec:Vq} that the stabilizer of $q_i = U_i \mathrm{C}_n$ is $\mathrm{U} \left( n \right)_{q_i} = U_i \mathrm{C}_n U_i^\dagger$.
We define the \textit{simultaneous stabilizer} of $ q_1, \ldots, q_k \in \mathcal{M}_n $ as the intersection $ \mathrm{U} \left( n \right)_{q_1, \ldots, q_k} \defeq \bigcap_{i=1}^k \mathrm{U} \left( n \right)_{q_i} $. As a subgroup of $ \mathrm{U} \left( n \right) $, the simultaneous stabilizer acts on  $ \mathcal{M}_n $ (the restriction of the $ \mathrm{U} \left( n \right) $-action \eqref{Un_action}). The following theorem characterizes equivalent MUB lists via orbits of this action.
\begin{theorem}\label{thm:main}
	If $ \left( q_1, \ldots, q_{k-1} \right) $ is a MUB list, then $ \mathcal{N} \left( q_1, \ldots, q_{k-1} \right) $ is closed under the action of the simultaneous stabilizer  $ \mathrm{U} \left( n \right)_{q_1, \ldots, q_{k-1} } $.
	Moreover, for any $p,r \in \mathcal{N} \left( q_1, \ldots, q_{k-1} \right) $, the two MUB lists $ \left( q_1, \ldots, q_{k-1}, p \right) $ and $ \left( q_1, \ldots, q_{k-1}, r \right) $ are equivalent iff $p$ and $r$ belong to the same $\mathrm{U} \left( n \right)_{q_1, \ldots, q_{k-1} }$-orbit in $ \mathcal{N} \left( q_1, \ldots, q_{k-1} \right) $.
\end{theorem}
\begin{proof}
	First, recall from \Cref{sec:Vq} that $ \mathcal{V} \left( q_i \right) $ is closed under the action of $ \mathrm{U} \left( n \right)_{q_i} $. This implies the intersection $\mathcal{N} \left( q_1, \ldots, q_{k-1} \right) = \bigcap_{i=1}^{k-1} \mathcal{V} \left( q_i \right) $ is preserved by each stabilizer $ \mathrm{U} \left( n \right)_{q_i} $, hence by $\mathrm{U} \left( n \right)_{q_1, \ldots, q_{k-1} }$. Thus, indeed there is a well-defined $\mathrm{U} \left( n \right)_{q_1, \ldots, q_{k-1} }$-action on $ \mathcal{N} \left( q_1, \ldots, q_{k-1} \right) $.

	By definition, the two lists $ \left( q_1, \ldots, q_{k-1}, p \right) $ and $ \left( q_1, \ldots, q_{k-1}, r \right) $ are equivalent iff there exists a unitary $U$ that satisfies both conditions: 
	\begin{itemize}
		\item $ U \cdot q_i = q_i $ for all $1 \leq i \leq k-1$;
		\item $U \cdot p = r$.
	\end{itemize}
	$U$ obeys the first condition iff it stabilizes all $q_i$, which occurs iff $ U \in \mathrm{U} \left( n \right)_{q_1, \ldots, q_{k-1} } $.
	By the definition of orbit, there exists an element $U \in \mathrm{U} \left( n \right)_{q_1, \ldots, q_{k-1} }$ that maps $p$ to $r$ iff these two points belong to the same $\mathrm{U} \left( n \right)_{q_1, \ldots, q_{k-1} }$-orbit. This completes the proof.
\end{proof}
In the context of our procedure, the above theorem completely characterizes the inequivalent choices for the $k$th basis.
\begin{remark*}
	Our theorem characterizes ordered MUB lists. However, when considering \textit{MUB sets}, the order of MUBs should be irrelevant. Therefore, two inequivalent MUB lists may correspond to equivalent MUB sets; this can occur, for example, if one is obtained from the other by reordering.
\end{remark*}

\subsection{The orbit structure of $ \mathcal{N} \left( q_1, \ldots, q_{k} \right) $}\label{sec:orbit_structure}
As pointed out earlier, the subsets $ \mathcal{N} \left( q_1, \ldots, q_{k} \right) $ may fail to be submanifolds of $\mathcal{M}_n$. This fact poses difficulties in analyzing our procedure; otherwise, we could have hoped to compute the dimension of every $ \mathcal{N} \left( q_1, \ldots, q_{k} \right) $, and use this data to find the maximal number of steps (hence the maximal MUB set). But if $ \mathcal{N} \left( q_1, \ldots, q_{k} \right) $ is not a differentiable manifold (i.e. it is singular), then its dimension is not even well-defined. In this subsection, we study the decomposition of $ \mathcal{N} \left( q_1, \ldots, q_{k} \right) $ into $\mathrm{U} \left( n \right)_{q_1, \ldots, q_k }$-orbits, where $\mathrm{U} \left( n \right)_{q_1, \ldots, q_k }$ is the simultaneous stabilizer defined in \Cref{sec:redundancies}. We show that this decomposition corresponds to a \textit{stratification} of $ \mathcal{N} \left( q_1, \ldots, q_{k} \right) $.
We start with the familiar case of $k=1$, corresponding to the classification of complex Hadamard matrices. Then we mention the general case. 
As our theorem suggests, an explicit description of the stratification in the general case can be used to simplify the construction of MUB lists.

We start with the simplest case: $k=1$ and $q_1 = e$ is the standard basis; so $ \mathcal{N} \left( q_1 \right) = \mathcal{V} \left( e \right) $.
For some values of $n$ (e.g. $2,3,5$), \textit{all} complex Hadamard matrices are equivalent; that is, the $ \mathrm{C}_n $ action on $ \mathcal{V} \left( e \right) $ is transitive, thus making it a homogeneous $ \mathrm{C}_n $-space. More generally, $ \mathcal{V} \left( e \right) $ decomposes as a disjoint union of double cosets:
\begin{equation}\label{double_cosets}
	\mathcal{V} \left( e \right) = \bigsqcup_{i \in \mathcal{I}}  \left( \mathrm{C}_n H_i \mathrm{C}_n \right) / \mathrm{C}_n ,
\end{equation}
where the quotient should be understood as referring only to the right $ \mathrm{C}_n $ factor. Another way of writing this down is:
\begin{equation}\label{classify_Hadamard}
	\mathcal{V} \left( e \right) = \bigsqcup_{i \in \mathcal{I}} \mathrm{C}_n \cdot q_i ,
\end{equation}
where each element in the disjoint union is the orbit (under the left $\mathrm{C}_n$ action) of the point $q_i \defeq H_i \mathrm{C}_n \in \mathcal{V} \left( e \right) $. Each $H_i$ is a representative of a double coset, or equivalently -- a representative of an equivalence class of complex Hadamard matrices.
Hence, finding the partition \eqref{classify_Hadamard} is the same as classifying the $n \times n$ complex Hadamard matrices up to equivalence. The set $ \mathcal{I} $ can be quite general: for $n \in \left\{ 2,3,5 \right\}$ it is a single point; but generally, it may be a union of isolated points and families parameterized by multiple real coordinates. This classification problem is unsolved for $ n \geq 6 $~\cite{dita2004some,tadej2006concise}. Note that it would be good to have a canonical choice for the representative $H_i$ of an arbitrary double coset.
We can always choose $H_i$ to be \textit{dephased}, i.e. with all entries in the first row and first column equal $1$. However, we still have the freedom to reorder the $n-1$ last rows and columns. Hence, there is more than one dephased representative in each double coset -- but there is only a finite number of those.

Recall that $\mathcal{M}_n$ is a homogeneous space, so there is nothing special about the standard basis $e = \mathds{1} \mathrm{C}_n$. For any $q = U \mathrm{C}_n \in \mathcal{M}_n$ 
we have an equivalent decomposition:
\begin{equation}\label{V_decomp}
	\mathcal{V} \left( q \right) = U \cdot \mathcal{V} \left( e \right) = \bigsqcup_{i \in \mathcal{I}} \left( U \mathrm{C}_n U^\dagger \right) \cdot \left( U \cdot q_i \right) ,
\end{equation}
where $ U \cdot q_i = U H_i \mathrm{C}_n $.
Note that the orbits now correspond to the left action of $ U \mathrm{C}_n U^\dagger $ (the stabilizer of $ q $).
Recall a unitary matrix $ V $ represents a $U\mathrm{C}_n$-unbiased basis iff $ H = U^\dagger V $ is a complex Hadamard matrix; thus, $UH$ indeed represents a point in $ \mathcal{V} \left( q \right) $ (where $H$ is an arbitrary Hadamard matrix).

$ \mathcal{V} \left( q \right) $ is a stratified space, with the stratification given by \eqref{V_decomp}.
Each component of \eqref{V_decomp} is an orbit of the compact Lie group $ \left( U \mathrm{C}_n U^\dagger \right) $, acting smoothly on the smooth manifold $ \mathcal{M}_n $. This ensures that every component is an embedded closed submanifold. By the same arguments, the $ \mathrm{U} \left( n \right)_{q_1, \ldots, q_k}$-orbits define a stratification of $ \mathcal{N} \left( q_1, \ldots, q_k \right) $.

We now show that for $k>1$, the $ \mathrm{U} \left( n \right)_{q_1, \ldots, q_k}$-orbits are all discrete.
We start with the following proposition:
\begin{proposition}\label[proposition]{prop:Cn_and_conjugate}
	For any complex Hadamard matrix $H$, the connected component of $ \mathrm{C}_n \cap H \mathrm{C}_n H^\dagger $ that contains the identity is $ Z \left( \mathrm{U} \left( n \right) \right) \cong \mathrm{U} \left( 1 \right) $, i.e. the center of $ \mathrm{U} \left( n \right) $, comprising precisely the scalar unitary matrices.
\end{proposition}
\begin{proof}
	As we have already noted, $ \mathrm{C}_n $ has $n!$ connected components, and the connected component of the identity is the torus $ \mathrm{T}^n $. Similarly, the connected component of the identity of $ H \mathrm{C}_n H^\dagger $ is the conjugated torus $ H \mathrm{T}^n H^\dagger $. Thus, the connected component of the identity in the intersection $ \mathrm{C}_n \cap H \mathrm{C}_n H^\dagger $ is $ \mathrm{T}^n \cap H \mathrm{T}^n H^\dagger $, which we now show to equal $ Z \left( \mathrm{U} \left( n \right) \right) $.
	
	Note that $ A \in \mathrm{T}^n \cap H \mathrm{T}^n H^\dagger $ iff both $A$ and $ H^\dagger A H $ are diagonal unitary matrices.
	For all $ 1 \leq i \leq n $ and any diagonal $A$, we have:
	\begin{align}
		\left[ H^\dagger A H \right]_{ii} & = \sum_{j=1}^n \left[ H^\dagger \right]_{ij} \left[ AH \right]_{ji} = \sum_{j=1}^n h_{ji}^* a_{jj} h_{ji} = \sum_{j=1}^n \abs{h_{ji}}^2 a_{jj} = \sum_{j=1}^n \frac{1}{n} a_{jj} = \nonumber\\
		& = \frac{1}{n} \Tr \left( A \right) ,
	\end{align}
	where we have used the fact the $H$ is a complex Hadamard matrix, hence $ \abs{h_{ji}}^2 = \frac{1}{n} $ for all $i,j$. Thus, the diagonal entries of $H^\dagger A H$ are all equal, implying that $A$ and $H^\dagger A H$ are both diagonal iff $H^\dagger A H$ is a scalar matrix, which of course implies that $ H^\dagger A H = A $.
\end{proof}
Note the proof has the following corollary:
\begin{equation}
	\mathrm{Core}_{\mathrm{U} \left( n \right)} \left( \mathrm{T}^n \right) \defeq \bigcap_{V \in \mathrm{U} \left( n \right)} V \mathrm{T}^n V^\dagger = Z \left( \mathrm{U} \left( n \right) \right) ,
\end{equation}
where $ \mathrm{Core}_{\mathrm{U} \left( n \right)} \left( \mathrm{T}^n \right) $ is the largest subgroup of $ \mathrm{T}^n $ which is normal in $ \mathrm{U} \left( n \right) $. The intersection $ \bigcap_{V \in \mathrm{U} \left( n \right)} V \mathrm{T}^n V^\dagger $ over \textit{all} unitaries $V$ equals the intersection $ \mathrm{T}^n \cap H \mathrm{T}^n H^\dagger $, for any Hadamard matrix $H$. Intuitively, this indicates that Hadamard matrices are the strongest obstructions for $ \mathrm{T}^n $ being normal (in $ \mathrm{U} \left( n \right) $). We also note that the normalizer of $ \mathrm{T}^n $ is $ \mathrm{C}_n $, the monomial matrices. In contrast, Hadamard matrices are the least monomial unitaries, in the sense that all entries have the same absolute value.

Now, consider a MUB pair $ \left( e, q \right) = \left( \mathds{1} \mathrm{C}_n, H \mathrm{C}_n \right) $ where $H$ is a Hadamard matrix (by \Cref{lemma:MUB_list_standard_form}, any MUB pair is equivalent to one of this form). We would like to show that the $ \mathrm{U} \left( n \right)_{e, q} $-orbits in $ \mathcal{N} \left( e, q \right) $ are discrete. Any such orbit is diffeomorphic to the quotient $ \mathrm{U} \left( n \right)_{e, q} / \mathrm{U} \left( n \right)_{e, q, p} $ for $ p \in \mathcal{N} \left( e, q \right) $.
Recalling $ \mathrm{U} \left( n \right)_{e, q} = \mathrm{C}_n \cap H \mathrm{C}_n H^\dagger $, the above proposition implies that $ \mathrm{U} \left( n \right)_{e, q} / Z \left( \mathrm{U} \left( n \right) \right) $ is a finite group ($ \mathrm{U} \left( n \right)_{e, q} $ has finitely many connected components since the same holds for $ \mathrm{C}_n $ and $ H \mathrm{C}_n H^\dagger $). 
By the discussion above, the center is contained in \textit{every} stabilizer, hence $ \mathrm{U} \left( n \right)_{e, q, p} \supseteq Z \left( \mathrm{U} \left( n \right) \right) $. Therefore, the orbit $ \mathrm{U} \left( n \right)_{e, q} / \mathrm{U} \left( n \right)_{e, q, p} $ is contained in the aforementioned finite group. The same reasoning holds for the $ \mathrm{U} \left( n \right)_{q_1, \ldots, q_k} $-orbits of $ \mathcal{N} \left( q_1, \ldots, q_k \right) $, for any $k > 1$.
\begin{remark*}
	Since the center $ Z \left( \mathrm{U} \left( n \right) \right) $ is the intersection of all stabilizers $ \mathrm{U} \left( n \right)_p $, the $ \mathrm{U} \left( n \right) $-action on $ \mathcal{M}_n $ induces a faithful $ \mathrm{PU} \left( n \right) $-action on $ \mathcal{M}_n $, where $ \mathrm{PU} \left( n \right) \defeq \mathrm{U} \left( n \right) / Z \left( \mathrm{U} \left( n \right) \right) $.
\end{remark*}

\section{Application: new equivalences of MUB triples in dimension $n=4$}\label{sec:app}
In this section, we apply our results in dimension $n=4$. For a fixed MUB pair $ \left( e, f_0 \right) $, we decompose $ \mathcal{N} \left( e, f_0 \right) $ into $ \mathrm{U} \left( n \right)_{e, f_0} $-orbits and discover new equivalences between MUB triples. We begin by rephrasing known results via our geometric perspective, and then provide a detailed explanation of a general method for computing $ \mathrm{U} \left( n \right)_{e, f_0} $ and its orbits.

In dimension $4$, the classification of complex Hadamard matrices is known~\cite{brierley2009all}. Every $ 4 \times 4 $ Hadamard matrix is equivalent to one of the following form:
	\begin{equation}
		F_4 \left( x \right) = \frac{1}{2} \begin{pmatrix}
			1 & 1 & 1 & 1 \\
			1 & 1 & -1 & -1 \\
			1 & -1 & i e^{i x} & -i e^{i x} \\
			1 & -1 & -i e^{i x} & i e^{i x}
		\end{pmatrix} , \quad x \in \left[ 0, \pi \right] .
	\end{equation}
	Let us recast this fact in our terms: the set of $ 4 \times 4 $ complex Hadamard matrices is $ \bigsqcup_{x \in \left[ 0, \pi \right] } \mathrm{C}_4 F_4 \left( x \right) \mathrm{C}_4 $, where $ \mathrm{C}_4 F_4 \left( x \right) \mathrm{C}_4 $ is the double coset of $ F_4 \left( x \right) $ under the two $ \mathrm{C}_4 $-actions (left and right multiplication).
	Equivalently, the subset $ \mathcal{V} \left( e \right) \subseteq \mathcal{M}_4 $ has the following decomposition into orbits of the left $ \mathrm{C}_4 $-action:
	\begin{equation}
		\mathcal{V} \left( e \right) \cong \bigsqcup_{x \in \left[ 0, \pi \right] } \mathrm{C}_4 \cdot f_x , \quad f_x \defeq F_4 \left( x \right) \mathrm{C}_4 ,
	\end{equation}
	so in this case the index set $ \mathcal{I} $ from \eqref{classify_Hadamard} is the interval $ \left[ 0, \pi \right]$, and the representative $ H_x $ is denoted $ F_4 \left( x \right) $.
	
	Now, consider the MUB pair $ \left( e, f_0 \right) $, where $ f_0 = F_4 \left( 0 \right) \mathrm{C}_4 $. From \cite{brierley2009all}, the only dephased unitary matrices which are mutually-unbiased to both $ \mathds{1} $ and $ F_4 \left( 0 \right) $ are of the form:
	\begin{equation}\label{H4_yz}
		H_4 \left( y, z \right) = \frac{1}{2} \begin{pmatrix}
			1 & 1 & 1 & 1 \\
			1 & 1 & -1 & -1 \\
			-e^{i y} & e^{i y} & e^{i z} & -e^{i z} \\
			e^{i y} & -e^{i y} & e^{i z} & -e^{i z}
		\end{pmatrix} , \quad y, z \in \left[ 0, \pi \right) .
	\end{equation}
	For each $y, z$ we obtain a distinct coset $ h_{y,z} \defeq H_4 \mathrm{C}_4 \left( y, z \right) $, hence $ \mathcal{N} \left( e, f_0 \right) = \left\{ h_{y,z} \mid y, z \in \left[ 0, \pi \right) \right\} $. We would like to know which (if any) of the MUB triples $ \left( e, f_0, h_{y,z} \right) $ are equivalent. By \Cref{thm:main}, $ h_{y, z} $ and $ h_{y', z'} $ yield equivalent MUB triples iff they lie in the same orbit of the group $ \tilde{G} \defeq \mathrm{C}_4 \cap F_4 \left( 0 \right) \mathrm{C}_4 F_4 \left( 0 \right)^\dagger $. 
	By \Cref{prop:Cn_and_conjugate} we know that $ G \defeq \tilde{G} / Z \left( \mathrm{U} \left( n \right) \right) $ is a finite group, and by the above remark, it suffices to characterize $ \mathcal{N} \left( e, f_0 \right) $ up to $G$-orbits. We compute $G$ explicitly via the following procedure. For $ \rho \in \mathrm{S}_4 $, let $ \mathrm{C}_4^\rho $ denote the corresponding connected component of $ \mathrm{C}_4 $, i.e.,
\begin{equation}\label{C4_conn_comp}
		\mathrm{C}_4^\rho \defeq \left\{ U \in \mathrm{U} \left( 4 \right) \mid U_{ij} = 0 \quad \forall i, j \; \mathrm{s.t.} \; j \neq \rho \left( i \right) \right\} .
\end{equation}
If $ U \in \tilde{G} $, there exist $ \rho, \sigma \in \mathrm{S}_4 $ such that $ U \in \mathrm{C}_4^\rho $ and $ F_4 \left( 0 \right)^\dagger U F_4 \left( 0 \right) \in \mathrm{C}_4^\sigma $. Let $ u_i \defeq U_{i, \rho \left( i \right)} $ denote the nonzero entries of $U$. Since $ F_4 \left( 0 \right)^\dagger U F_4 \left( 0 \right) \in \mathrm{C}_4^\sigma $, we obtain from \eqref{C4_conn_comp} the following homogeneous linear equations for $ u_i $:
\begin{align}
		& \left[ F_4 \left( 0 \right)^\dagger U F_4 \left( 0 \right) \right]_{kl} = 0 && \forall k, l \; \mathrm{s.t.} \; l \neq \sigma \left( k \right) \nonumber\\
		\Leftrightarrow \quad & \sum_{m=1}^{4} \left[ F_4 \left( 0 \right)^\dagger \right]_{j, \rho^{-1} \left( m \right)} \left[  F_4 \left( 0 \right) \right]_{m,l} u_{\rho^{-1} \left( m \right)} = 0 && \forall k, l \; \mathrm{s.t.} \; l \neq \sigma \left( k \right) .
\end{align}
It is straightforward to go over all pairs $ \rho, \sigma \in \mathrm{S}_4 $ and seek all nontrivial solutions to this system of equations. The solution space turns out either trivial or one-dimensional; and in one-dimensional cases the $u_i$ always have equal modulus, and by normalizing we obtain a unique element of $G$.
It turns out that only $8$ out of the $24$ permutations participate in nontrivial intersections; out of those $8$ permutations, each one intersects $4$ (including itself); thus, $\abs{G} = 32$. Moreover, $G$ is given by a semidirect product $ G \cong B \ltimes \mathbb{Z}_4 $, where $ B \subseteq \mathrm{S}_4 $ is given by:
\begin{equation}
	B \defeq \left\{ \mathrm{id}, \left( 12 \right), \left( 3 4 \right), \left( 1 2 \right) \left( 3 4 \right) , \left( 1 3 \right) \left( 2 4 \right) , \left( 1 3 2 4 \right) , \left( 1 4 2 3 \right) , \left( 1 4 \right) \left( 2 3 \right) \right\} ,
\end{equation}
and is generated by $\left( 1 2 \right)$, $ \left( 3 4 \right) $ and $ \left( 1 3 \right) \left( 2 4 \right) $. $ \mathbb{Z}_4 $ is realized by diagonal matrices $ D_a \in \mathrm{T}^4 $. The diagonal entries of each $D_a$ are given by another row of $ F_4 \left( 0 \right) $:
\begin{equation}
	D_0 = \mathds{1} , \quad D_1 = \mathrm{diag} \begin{pmatrix}
		1 \\ -1 \\ i \\ -i
	\end{pmatrix} , \quad  D_2 = \mathrm{diag} \begin{pmatrix}
		1 \\ 1 \\ -1 \\ -1
	\end{pmatrix} , \quad
	D_3 = \mathrm{diag} \begin{pmatrix}
	1 \\ -1 \\ -i \\ i
\end{pmatrix} .
\end{equation}
One easily verifies that $ D_a D_b = D_{a+b \, \mathrm{mod} \, 4} $. Letting $ R_\pi $ denote the permutation matrix corresponding to $\pi \in B $, we can provide a complete description of $G$:
all products $ R_\pi D_a $, up to a scalar matrix $ e^{i \theta} \mathds{1} \in \mathrm{Z} \left( \mathrm{U} \left( n \right) \right) $. Equivalently, the matrices $ R_\pi D_a $ form a \textit{projective representation} of $G$; for example, $ D_1 R_{\left( 13 \right) \left( 24 \right)} D_1 = i R_{\left( 13 \right) \left( 24 \right)} $ (we shall use this fact later).

Now we can compute all the left actions $ R_\pi D_a \cdot h_{y,z} $, and test whether any two distinct $ h_{y,z} $ are in the same $G$-orbit. Starting with $\mathbb{Z}_4$, we find:
\begin{equation}
	D_1 H_4 \left( y, z \right) = \frac{1}{2} \begin{pmatrix}
		1 & 1 & 1 & 1 \\
		-1 & -1 & 1 & 1 \\
		-i e^{i y} & i e^{i y} & i e^{i z} & -i e^{i z} \\
		-i e^{i y} & i e^{i y} & -i e^{i z} & i e^{i z}
	\end{pmatrix} .
\end{equation}
This matrix is not in the form of the canonical representative \eqref{H4_yz}, but we can bring it to this form using the \textit{right} $ \mathrm{C}_n $-action; and since the first row is already dephased, we need only consider column permutations. To get the second row in the correct form we must swap the first two columns with the second two, but how do we know which way to do it? Note that $ \left[ H_4 \left( y, z \right) \right]_{32} $ is either $1$, or has strictly positive imaginary part. If $ 0 \leq z < \frac{\pi}{2} $, then $ i e^{i z} = e^{i \left( z + \frac{\pi}{2} \right)} $ obeys this condition; otherwise, $ -i e^{i z} = e^{i \left( z - \frac{\pi}{2} \right)} $ obeys this condition. An analogous statement holds for $ \pm i e^{i y} $, based on the range of $y$. Denote:
\begin{equation}
	y' \defeq \begin{cases}
		z + \frac{\pi}{2} ; & 0 \leq z < \frac{\pi}{2} \\
		z - \frac{\pi}{2} ; & \frac{\pi}{2} < z < \pi
	\end{cases} , \quad z' \defeq \begin{cases}
	y + \frac{\pi}{2} ; & 0 \leq y < \frac{\pi}{2} \\
	y - \frac{\pi}{2} ; & \frac{\pi}{2} < y < \pi
\end{cases} .
\end{equation}
By choosing a column permutation suitable to the ranges of $y,z$, we obtain a matrix $ H_4 \left( y', z' \right) $. Thus, we have found $ D_1 \cdot h_{y,z} = h_{z + \frac{\pi}{2} \, \mathrm{mod} \, \pi, \; y + \frac{\pi}{2} \, \mathrm{mod} \, \pi} $.
Moving forward, we find:
\begin{equation}
	D_2 H_4 \left( y, z \right) = \frac{1}{2} \begin{pmatrix}
		1 & 1 & 1 & 1 \\
		1 & 1 & -1 & -1 \\
		e^{i y} & -e^{i y} & -e^{i z} & e^{i z} \\
		-e^{i y} & e^{i y} & -e^{i z} & e^{i z}
	\end{pmatrix} \xrightarrow{c_1 \leftrightarrow c_2, c_3 \leftrightarrow c_4} H_4 \left( y, z \right) ,
\end{equation}
by which we mean that after swapping the first column with second and the third column with the fourth, we obtain the matrix $ H_4 \left( y, z \right) $ again. Thus, $ D_2 $ stabilizes $ H_4 \left( y, z \right) $ for all $y,z$; hence $ D_2 \in \mathrm{U} \left( n \right)_{e, f_0, h_{y,z}} $. Since $ D_3 = D_1 D_2 $, we already know that the action of $D_3$ would be the same as that of $D_1$.

Now we can proceed to examine the action of $B$ on $ \mathcal{N} \left( e, f_0 \right) $. We have:
\begin{equation}
	R_{\left( 1 2 \right)} H_4 \left( y, z \right) = \frac{1}{2} \begin{pmatrix}
		1 & 1 & -1 & -1 \\
		1 & 1 & 1 & 1 \\
		-e^{i y} & e^{i y} & e^{i z} & -e^{i z} \\
		e^{i y} & -e^{i y} & e^{i z} & -e^{i z}
	\end{pmatrix} \xrightarrow{\mathrm{dephasing}} \frac{1}{2} \begin{pmatrix}
	1 & 1 & 1 & 1 \\
	1 & 1 & -1 & -1 \\
	-e^{i y} & e^{i y} & -e^{i z} & e^{i z} \\
	e^{i y} & -e^{i y} & -e^{i z} & e^{i z}
\end{pmatrix} ,
\end{equation}
where the dephasing step multiplies the two rightmost columns by $-1$. By swapping the third and fourth columns of the matrix we obtained above, we get $ H_4 \left( y, z \right) $; hence, $ R_{\left( 1 2 \right)} $ is in the stabilizer. We can similarly verify that $ R_{\left( 3 4 \right)} $ belongs to the stabilizer as well. Finally, we should examine the action of $ R_{\left( 1 3 \right) \left( 2 4 \right)} $; for the sake of convenience, we first consider $ R_{\left( 1 3 \right) \left( 2 4 \right)} D_1 $:
\begin{equation}
	R_{\left( 1 3 \right) \left( 2 4 \right)} D_1 H_4 \left( y, z \right)  \xrightarrow{\mathrm{dephasing}} \frac{1}{2} \begin{pmatrix}
	1 & 1 & 1 & 1 \\
	1 & 1 & -1 & -1 \\
	i e^{-i y} & -i e^{-i y} & -i e^{-i z} & i e^{-i z} \\
	-i e^{-i y} & i e^{-i y} & -i e^{-i z} & i e^{-i z}
\end{pmatrix} .
\end{equation}
Recall we want the entry in the third row and first column to be either $-1$ or have negative imaginary part, and note that $ i e^{-i y} = e^{i \left( \frac{\pi}{2} -y \right)} = -e^{i \left( \frac{3\pi}{2} -y \right)} $. Thus, if $ \frac{\pi}{2} < y < \pi $ this condition is satisfied, and we can define $ y' = \frac{3\pi}{2} -y $; otherwise we should swap the first two columns and define $ y' = \frac{\pi}{2} -y $. Applying the same logic for $z$, we find that $ R_{\left( 1 3 \right) \left( 2 4 \right)} D_1 \cdot h_{y,z} = h_{ \left( \frac{3\pi}{2} -y \right) \mathrm{mod} \, \pi, \; \left( \frac{3\pi}{2} -z \right)   \mathrm{mod} \, \pi } $.

To summarize, we computed the $ \mathrm{U} \left( n \right)_{e, f_0} $-action on $h_{y,z} \in \mathcal{N} \left( e, f_0 \right) $, which is equivalent to the $G$-action (by the remark at the end of \Cref{sec:orbit_structure}). For each $y,z$, we found that the stabilizer $G_{h_{y,z}}$ of $h_{y,z}$ is generated by $D_2$, $R_{\left( 1 2 \right)}$ and $ R_{\left( 3 4 \right)} $. These are commuting elements of order two, hence $ G_{h_{y,z}} \cong \mathbb{Z}_2 \times \mathbb{Z}_2 \times \mathbb{Z}_2 $. The orbit of $h_{y,z}$ under $G$ is diffeomorphic to $ G / G_{h_{y,z}} $; thus, it has four elements:
\begin{equation}
	\begin{split}
		\mathds{1} \cdot h_{y,z} & = h_{y,z}, \\ 
		D_1 \cdot h_{y,z} & = h_{ \left( z + \frac{\pi}{2} \right) \mathrm{mod} \, \pi, \; \left( y + \frac{\pi}{2} \right) \mathrm{mod} \, \pi }, \\
		R_{\left( 1 3 \right) \left( 2 4 \right)} D_1 \cdot h_{y,z} & = h_{ \left( \frac{3\pi}{2} -y \right) \mathrm{mod} \, \pi, \; \left( \frac{3\pi}{2} -z \right) \mathrm{mod} \, \pi }, \\ 
		R_{\left( 1 3 \right) \left( 2 4 \right)} \cdot h_{y,z} & = h_{\left( \pi-z \right) \mathrm{mod} \, \pi , \; \left( \pi-y \right) \mathrm{mod} \, \pi } .
	\end{split}
\end{equation}
To compute $ R_{\left( 1 3 \right) \left( 2 4 \right)} \cdot h_{y,z} $ we used the fact $ D_1 R_{\left( 13 \right) \left( 24 \right)} D_1 = i R_{\left( 13 \right) \left( 24 \right)} $, which implies that $ R_{\left( 1 3 \right) \left( 2 4 \right)} \cdot h_{y,z} = D_1 \cdot \left( R_{\left( 1 3 \right) \left( 2 4 \right)} D_1 \cdot h_{y,z} \right) $.

Thus, we have established the following:
\begin{proposition}\label[proposition]{prop:MUB_triples_in_dim_4s}
	The MUB triples $ \left( e, f_0, h_{y,z} \right) $ and $ \left( e, f_0, h_{y', z'} \right) $ are equivalent (in the sense of \Cref{def:equiv_MUB_lists}) if and only if one of the following holds:
	\begin{itemize}
		\item $y' = y$ and $ z' = z$;
		\item $ y' = \left( z + \frac{\pi}{2} \right) \mathrm{mod} \, \pi $ and $ z' = \left( y + \frac{\pi}{2} \right) \mathrm{mod} \, \pi $;
		\item $ y' = \left( \frac{3\pi}{2} -y \right) \mathrm{mod} \, \pi $ and $ z' = \left( \frac{3\pi}{2} -z \right) \mathrm{mod} \, \pi $;
		\item $ y' = \left( \pi - z \right) \mathrm{mod} \, \pi $ and $ z' = \left( \pi -  y \right) \mathrm{mod} \, \pi $.
	\end{itemize}
\end{proposition}
\begin{proof}
	The proposition follows from \Cref{thm:main}, combined with our computation of the $ \mathrm{U} \left( n \right)_{e, f_0} $-orbit of $ h_{y,z} $.
\end{proof}
These equivalences were hitherto unknown (to our knowledge), and reduce the parameter space by a factor of $4$.

Note that the method we used in this section can be executed for any value of $n$. Moreover, we can classify MUB quadruples, quintuples etc. In fact, in the process of classifying MUB triples we already computed the simultaneous stabilizer $ U_{e, f_0, h_{y,z}} $ of a triple, which shall be required to classify quadruples. This is a general feature: when finding the $\mathcal{U} \left( n \right)_{q_1, \ldots, q_{k-1}}$-orbit of a point $ p \in \mathcal{N} \left( q_1, \ldots, q_{k-1} \right) $, we compute the stabilizer $ \mathcal{U} \left( n \right)_{q_1, \ldots, q_{k-1}, p} $.

\section{Conclusion}\label{sec:conclusion}
In this paper, we considered the task of constructing an ordered list of mutually unbiased bases. We introduced a geometric formulation of this task. To do so, we defined a manifold endowed with a metric and an isometric action of the unitary group. The points of the manifold correspond to orthonormal bases, and the geometric structure captures the notions of mutual-unbiasedness and equivalence.
We outlined a procedure for constructing MUB lists, where on the $k$th step one chooses a new basis, unbiased to all previously-chosen bases. We proved a theorem that characterizes in geometric terms choices that yield equivalent MUB lists.

Our results shed new light on the connection between mutually unbiased bases and complex Hadamard matrices. Our theorem explains the significance of classifying Hadamard matrices to the study of MUBs, and demonstrates that this classification is a special case of a more general decomposition. 
Moreover, we can reformulate the existence problem of mutually unbiased bases in geometric terms: a set of $m$ MUBs exists in dimension $n$, iff there exist points $q_1, \ldots, q_m \in \mathcal{M}_n$ such that $\forall i, \; q_i \in \mathcal{N} \left( q_1, \ldots, q_{i-1} \right)$.
Our geometric perspective clarifies existing classifications of mutually unbiased bases, and also generates new ones. It may help to (i) shrink parameter spaces dramatically, (ii) make full numerical MUB searches feasible, and (iii) unify scattered Hadamard families.

There are several promising directions in which our results could be extended. Recall that we have only considered (ordered) MUB lists; however, for many applications of MUBs, two MUB lists with the same elements appearing in a different order would be considered equivalent. Thus, future work may seek to generalize our results and characterize (unordered) MUB sets up to equivalence.

Our results may also open up the possibility of applying additional mathematical tools to the existence problem of MUBs. Indeed, the spaces we have defined are endowed with a rich geometric structure, of which we only utilized very little.
Smooth manifolds with compact Lie group actions have a rich theory, which may allow for a systematic study of the stratification of $ \mathcal{N} \left( q_1, \ldots, q_{i-1} \right) $. In particular, the geometric formulation may assist with the classification of Hadamard matrices in dimension $n=6$, where it is still open.
We also note that $ \mathcal{M}_n $ is closely related to the complete flag manifold $ \tilde{\mathcal{M}}_n $, which has important applications in representation theory. By the Borel-Weil theorem, each irreducible representation of $ \mathrm{U} \left( n \right) $ is given as the space of global sections of a certain holomorphic line bundle on $ \tilde{\mathcal{M}}_n $. In contrast, our $ \mathcal{M}_n $ is naturally equipped with a \textit{real} vector bundle of rank $n-1$. The cohomology groups (actually real vector spaces) of this vector bundle furnish representations of $ \mathrm{U} \left( n \right) $. There may be an unexplored connection between these representations, and those realized on the line bundles on $ \mathrm{U} \left( n \right) $.
Moreover, $ \tilde{\mathcal{M}}_n $ is a K\"{a}hler manifold, while $ \mathcal{M}_n $ inherits a Riemannian metric via its embedding in the real Grassmannian $ \mathrm{Gr} \left( n-1, \mathfrak{su} \left( n \right) \right) $. Curiously, the chordal distance on $ \mathcal{M}_n $ (also inherited from the Grassmannian) seems to be more directly pertinent to the problem of MUBs, compared to the richer geometric structures of the vector bundle and Riemannian metric. Future work may wish to seek connections between all of these structures and MUBs.
Finally, note the decomposition \eqref{double_cosets} of $ \mathcal{V} \left( q \right) $ into double cosets holds a striking resemblance to the \textit{Bruhat decomposition} of $ \mathcal{M}_n $.
These facts may hint towards the possibility of applying further representation-theoretic tools to study the geometric problems presented in this paper.

\backmatter

\bmhead{Acknowledgements}
We are grateful to Leonid Polterovich, Boris Kunyavskii, Joseph Bernstein and Mikhail Katz for helpful discussions.

\section*{Declarations}
\bmhead{Funding}
This work was partially supported by the European Union's Horizon Europe research and innovation programme under grant agreement No. 101178170 and by the Israel Science Foundation under grant agreement No. 2208/24.

\bmhead{Competing interests}
The authors have no relevant financial or non-financial interests to disclose.

\bmhead{Ethics approval}
Not applicable.

\bmhead{Consent for publication}
Not applicable.

\bmhead{Data availability}
No datasets were generated or analysed during the current study.

\bmhead{Materials availability}
Not applicable.

\bmhead{Code availability}
The code used in this study is publicly available at \url{https://github.com/smitke6/Equivalent-MUB-lists}.

\bmhead{Author contributions}
Both authors contributed to conceiving the project. The proofs and computations were performed by Amit Te'eni. The manuscript was written by Amit Te'eni with comments and edits provided by Eliahu Cohen.

\begin{appendices}




\end{appendices}



\end{document}